\documentclass[11pt]{article}

\title{A new approach to Cholesky-based covariance regularization in high dimensions}
\author{Adam J. Rothman,
        Elizaveta Levina, and Ji Zhu \\
        Technical Report \# 480 \\
Department of Statistics \\
University of Michigan}

\usepackage{fullpage}

\usepackage{amsmath, amssymb, amsfonts, amscd, xspace, pifont}
\usepackage{epsfig, psfrag, pstricks}
\usepackage{float, fancybox}
\usepackage{amsthm}
\usepackage{dsfont}
\usepackage{natbib}
\theoremstyle{plain}

\newtheorem{proposition}{Proposition}

\newcommand{\be}{\begin{equation}}
\newcommand{\ee}{\end{equation}}

\newcommand{\V}[1]{\ensuremath{\boldsymbol{#1}}\xspace}

\newcommand{\F}[1]{\ensuremath{\mathrm{#1}}\xspace}

\newcommand{\tr}{\F{trace}}

\newcommand{\I}{\mathds{1}}

\newcommand{\argmax}{\operatorname*{argmax}}
\newcommand{\argmin}{\operatorname*{argmin}}

\def\twoImages#1#2#3#4#5#6 
{
 \centerline{\hfill\makebox[#2]{#3}\hfill\makebox[#5]{#6}\hfill}
\smallskip
  \centerline{\hfill
    \includegraphics[width=#2]{#1}
    \hfill
    \includegraphics[width=#5]{#4}
    \hfill} 
}
\def\oneImage#1#2#3
{
  \centerline{\hfill
    \includegraphics[width=#2]{#1}
    \hfill
    \smallskip}
    \centerline{\hfill\makebox[#2]{#3}\hfill\medskip}
}

\begin{document}

\maketitle
\begin{abstract}
In this paper we propose a new regression interpretation of the Cholesky factor of the covariance matrix, as opposed to the well known regression interpretation of the Cholesky factor of the inverse covariance, which leads to a new class of regularized covariance estimators suitable for high-dimensional problems.  Regularizing the Cholesky factor of the covariance via this regression interpretation always results in a positive definite estimator. In particular, one can obtain a positive definite banded estimator of the covariance matrix at the same computational cost as the popular banded estimator proposed by \citet{bl06}, which is not guaranteed to be positive definite.    We also establish theoretical connections between banding Cholesky factors of the covariance matrix and its inverse and constrained maximum likelihood estimation under the banding constraint, and compare the numerical performance of several methods in simulations and on a sonar data example. 
\end{abstract}

\section{Introduction}
Statistical inference for high-dimensional data has become increasingly necessary in recent years.  Advances in computing have made
high-dimensional data analysis possible in a number of important applications, including spectroscopy, 
fMRI, text retrieval, gene arrays, climate studies, and imaging. 
Many multivariate data analysis techniques applied to high-dimensional data require an estimate of 
the covariance matrix or its inverse; however, traditional estimation by the sample covariance matrix is known to perform poorly when there are more variables than observations ($p > n$)  -- see 
\citet{johnstone01} and references therein for a detailed discussion.  
A number of alternative estimators have been proposed for high-dimensional problems,  many of which exploit various sparsity  
assumptions about the population covariance matrix or its inverse. 

The problems of estimating the covariance matrix and its inverse are usually considered separately in this context, since in high dimensions inversion is costly and not always accurate.  When the goal is to estimate the inverse covariance matrix, also known as the concentration matrix, a popular method is to add the lasso ($\ell_1$) penalty on the 
entries of the inverse covariance matrix to the normal likelihood \citep{banerjee06, yuan07, spice, fht08}, which has been extended to
more general penalties by \citet{lam07}.  Other estimators of the inverse
covariance exploit the assumption that variables have a natural ordering, and those far apart in the ordering have small partial correlations.  
These estimators usually rely on the modified Cholesky decomposition of the inverse covariance matrix (see details in Section \ref{sec:cholfactor},
since this decomposition has a nice regression interpretation and 
regression regularization can be applied; see 
 \citet{wu03},  \citet{huang06},
\citet{bl06}, and \citet{levina_zhu06}.  

If the covariance matrix (rather than its inverse) is of interest, a simple way to improve on the sample covariance, both theoretically and in practice, 
is to threshold small elements to zero \citep{bl07, elkaroui07,genthresh}.  
Under the assumption that variables are ordered and those 
far apart in the ordering are only weakly correlated, a better option is to band or taper the sample covariance matrix \citep{bickel03, furrer06, bl06, cai08}.  These simple approaches are attractive 
for problems in very high dimensions since they have a small computational cost; however,
these estimators are not generally guaranteed to be positive definite, although some forms of tapering can
guarantee positive semi-definite estimates.  Alternatively, a positive definite constrained maximum likelihood estimator can be computed under the constraint enforcing any given 
pattern of zeros \citep{chaudhuri07}, but this algorithm is only applicable when there are fewer variables than observations ($p < n$).  

In this paper we  show that the modified Cholesky factor of the
covariance matrix (rather than its inverse) also has a natural 
regression interpretation, and therefore all Cholesky-based 
regularization methods can be applied to the covariance matrix itself instead of its inverse to obtain a sparse estimator with guaranteed positive definiteness. As with all Cholesky-based regularization methods, this approach 
exploits the assumption
of naturally ordered variables where variables far apart in the ordering tend to have small correlations.    The simplest estimator in this new class
is banding the covariance Cholesky factor. Unlike banding the covariance matrix itself, it is guaranteed to be positive definite, 
but still has the same low computational complexity.

The rest of this paper is organized as follows:  
we discuss the modified Cholesky factorization of the covariance matrix and 
its regression interpretation  in Section \ref{sec:cholfactor}. 
Regularization techniques appropriate for the Cholesky factor of covariance are described in Section \ref{sec:regchol}. In addition, we connect sparsity 
in the covariance matrix to sparsity in its Cholesky factor and use this to contrast the maximum likelihood properties of banding the Cholesky factor of 
covariance and banding the Cholesky factor of the inverse.  In particular, 
we prove that Cholesky banding of the inverse is the constrained maximum 
likelihood estimator for normal data under the constraint 
that the inverse covariance matrix is banded. Numerical performance of regularized Cholesky-based estimators of the covariance is illustrated both on simulated data (Section \ref{sec:sim}) and on a spectroscopy data example (Section \ref{sec:realdata}).  

\section{Modified Cholesky decomposition of the covariance matrix} \label{sec:cholfactor}

Throughout the paper we assume that the data $\V{X}_1, \ldots, \V{X}_n$
 are independent and identically distributed $p$-variate random vectors with population covariance matrix $\Sigma$ and, without loss of generality, mean $\V{0}$.  Let $\hat\Sigma$ denote the sample covariance matrix (the maximum likelihood version), 
 $$
 \hat\Sigma = \frac{1}{n}\sum_{i=1}^{n} (\V{X}_i - \V{\bar X})(\V{X}_i - \V{\bar X})^{T}.
 $$
 
As a tool for regularizing the inverse covariance matrix, \citet{pourahmadi99} suggested using the modified Cholesky factorization of $\Sigma^{-1}$.  This factorization arises from regressing each variable $X_j$ on $X_{j-1}, \dots, X_1$ for $2\leq j \leq p$.  Fitting regressions
  $$
  X_j   =  \sum^{j-1}_{q=1} {(-t_{jq})X_q} + \epsilon_j  = \hat X_j + \epsilon_j \ ,
  $$
let $\epsilon_j$ denote the error term in regression $j$, $j \ge 2$, and let $\epsilon_1 = X_1$. Let 
$D =  \mbox{var}(\V{\epsilon})$ be the diagonal matrix of error variances and 
$T = [t_{jq}]$ the lower-triangular matrix containing 
regression coefficients (with the opposite sign), with ones on the diagonal.  
Then writing $\V{\epsilon} = \V{X} - \V{\hat{X}} = T \V{X}$
and using the independence of errors we have, 
$$
D  =  \mbox{var}(\V{\epsilon}) =  \mbox{var}(T \V{X})= T \Sigma T^T
$$
and thus 
\begin{equation} \label{invchol}
\Sigma^{-1} = T^T D^{-1} T.
\end{equation}
This decomposition transforms inverse 
covariance matrix estimation into a regression problem, and hence regularization approaches for regression can be applied.   In practice, the coefficients are computed by regressing each variable $X_j$ on its predecessors $X_1, \dots, X_{j-1}$ (after centering all the variables).  If these regressions are not regularized, the resulting estimate is simply $\hat\Sigma^{-1}$.   {\em Banding} the Cholesky factor of the inverse refers to regularizing by only including the immediate $k$ predecessors in the regression, $X_{\max(1,j-k)}, \dots, X_{j-1}$, for some fixed $k$ \citep{wu03, bl06}.

The modified Cholesky factorization of $\Sigma$ can be obtained by simply inverting (\ref{invchol}).  Let $L=T^{-1}$ and rewrite  $\V{X} = L \V{\epsilon}$. 
Then,
\begin{equation} \label{chol}
\Sigma = {\rm var} \ (L \V{\epsilon}) = LDL^T \ .
\end{equation}
Our main interest here is in the regression interpretation of this decomposition.  By analogy to (\ref{invchol}), we can interpret (\ref{chol}) as resulting from a new sequence of regressions, where each variable $X_j$ is regressed on 
all the previous regression {\em errors} $\epsilon_{j-1}, \ldots, \epsilon_{1}$ (rather than the variables themselves).  For $j \geq 2$,  we have the sequence
of regressions,
\begin{equation} \label{cholregs}
X_j = \sum_{q=1}^{j-1} { l_{jq} \epsilon_{q}} + \epsilon_j = \tilde X_j + \epsilon_j \ .  
\end{equation}

The decompositions above apply to the population matrices.  Let 
$X = [\V{x}_1, \cdots, \V{x}_p]$ be 
  the $n\times p$ data matrix, where each column $\V{x}_j \in \mathbb{R}^{n}$ 
  is already centered by its sample mean. For the first variable, we set 
  $\V{e}_1  = \V{x}_1$.   
  For $2 \leq j \leq p$, let $\V{l}_j = (l_{j1}, \ldots, l_{j,j-1})^{T}$,     
$Z_j  = [\V{e}_1,\ldots, \V{e}_{j-1} ]$, and compute coefficients and the residual, respectively, as  
  \begin{align}
    \V{\hat l}_j  & = \argmin_{\V{l}_j} \|\V{x}_j - Z_{j} \V{l}_j \|^2 \nonumber  \ , \\
     \V{e}_{j} & = \V{x}_j - Z_{j} \V{\hat l}_j \ . \label{regress}  
  \end{align}
The variances are estimated as 
$$\hat d_{jj}  = \frac{1}{n} \|  \V{e}_{j} \|^2 \ . 
$$
Let $Z = [\V{e}_1, \cdots, \V{e}_p]$
  denote the $n\times p$ matrix of residuals from carrying out the regressions 
  in (\ref{cholregs}) sequentially.
  Here we assume that $p < n$ to ensure that all model matrices are of full column rank; Section \ref{sec:regchol} discusses the rank deficient case when $p \geq n$.  Performing the regressions in (\ref{regress}) amounts to, for each $j \geq 2$, orthogonally projecting  
  the response $\V{x}_j$ onto  the span of $\V{e}_1, \cdots, \V{e}_{j-1}$ to estimate $\V{\hat l}_{j}$.   After the last projection we have an
  orthogonal basis $\{\V{e}_1, \cdots, \V{e}_p \}$, and the estimates $\hat L$ and $\hat D$.   This algorithm is nothing but  
a scaled version of Gram-Schmidt orthogonalization of the data matrix $X$ 
for computing its QR decomposition, where the upper triangular matrix 
$R$ is restricted to have positive diagonal entries. The
  orthonormal matrix $Q$ is the matrix $Z$ with its column vectors scaled to 
  have unit length and $R^{T} =  \hat L {(n\hat D)}^{\frac{1}{2}}$.  If all regressions are fitted without any regularization, simply by least squares (as described above), the resulting estimate recovers the sample covariance matrix: 
  $$
  \hat\Sigma = \frac{1}{n} X^{T} X = \frac{1}{n} R^{T} R = \hat L \hat D \hat L^{T} \ .
  $$

\section{Regularized estimation of the Cholesky factor $L$} \label{sec:regchol}

It is clear that in order to improve on the sample covariance, the regressions in (\ref{regress}) need to be regularized.  In this section we describe several 
estimators that introduce sparsity in covariance Cholesky factor $L$.
We also connect sparsity patterns in positive definite matrices with 
sparsity patterns in their Cholesky factors and use this 
to analyze the connection between banding Cholesky factors and 
constrained maximum likelihood estimation.  

\subsection{Banding the Cholesky factor}
The simplest way to introduce sparsity in the Cholesky factor $L$ is to 
estimate only the first $k$ sub-diagonals of $L$
and set the rest to zero.  This approach for banding the Cholesky factor of the inverse was proposed by 
\citet{wu03} and \citet{bl06}.  In practice, it means that each 
variable $\V{x}_j$ is regressed
 on the $k$ previous residuals $[\V{e}_{j-k}, \ldots, \V{e}_{j-1}]$, for all $j \ge 2$.  Note that the index $j-k$ everywhere is understood to mean $\max(1,j-k)$.  Let $\V{l}_{j}^{(k)} = (l_{j,j-k}, \ldots, l_{j,j-1})^{T}$ and    $Z_{j}^{(k)}  = [\V{e}_{j-k},\ldots, \V{e}_{j-1} ]$. Then 
we  compute,
  \begin{align}
    \V{\hat l}_{j}^{(k)}  & = \argmin_{\V{l}_{j}^{(k)}} \|\V{x}_j - Z_{j}^{(k)} \V{l}_{j}^{(k)} \|^2  \ , \label{bandreg}\\
     \V{e}_{j} & = \V{x}_j - Z_{j}^{(k)} \V{\hat l}_{j}^{(k)} \ .  \nonumber
  \end{align}

In each regression, the design matrix $Z_{j}^{(k)}$ has orthogonal columns, 
which allows (\ref{bandreg}) to be solved with at most $k$ univariate regressions.  Hence the computational cost of banding the Cholesky factor in this manner 
is $O(kpn)$, the same order as banding the sample covariance matrix without the Cholesky decomposition.  To ensure
that design matrices are of full rank, the banding parameter $k$ must be less than $\min (n-1, p)$.   Also note that while each design matrix $Z_{j}^{(k)}$
has orthogonal columns, all of the residual vectors $\V{e}_{1}, \ldots, \V{e}_p$ are 
not necessarily mutually orthogonal; $\V{e}_{j}$ and $\V{e}_{j'}$  are only guaranteed to be orthogonal if $| j - j'| \le k$.  
  
\subsection{Connection to constrained maximum likelihood} \label{sec:theory}
Given that a Cholesky-based banded estimator is always positive definite, it is natural to ask whether it coincides with the maximum likelihood estimator under the banded constraint.  Here we show that, somewhat surprisingly, the answer depends on whether the banding is applied to the Cholesky factor of the 
inverse or of the covariance matrix itself: the former estimator coincides with constrained maximum likelihood estimator, and the latter does not.  In order to show this, we first establish some relationships between zero patterns in positive definite matrices and their Cholesky factors.  

\begin{proposition}\label{propzero}
Given a positive definite matrix $\Sigma$ with modified Cholesky decomposition 
$\Sigma = LDL^{T}$, where $L$ is lower triangular, 
for any row $i$ and $c(i) < i$, $\sigma_{i1} = \cdots = \sigma_{i, c(i)} = 0$ if and only if $l_{i1} = \cdots = l_{i, c(i)} = 0$.
\end{proposition}

\begin{proof}
Using the expression
$$
\sigma_{ij} = \sum_{m=1}^{j} {l_{im}l_{jm}d_{mm}},
$$
it is obvious that $l_{i1} = \cdots = l_{i, c(i)} = 0$ implies 
$\sigma_{i1} = \cdots = \sigma_{i, c(i)} = 0$.

Now assume $\sigma_{i1} = \cdots = \sigma_{i, c(i)} = 0$ for some $i$.  
The sequential column-wise formula for computing the modified Cholesky factorization \citep{watkins91}, is given by, for $i > j$,   
\begin{align}
d_{ii} & = \sigma_{ii} - \sum_{m=1}^{i-1} {l_{im}^2 d_{mm}} \ , \nonumber \\
l_{ij} & = \frac{1}{d_{jj}} \left(\sigma_{ij} - \sum_{m=1}^{j-1}{l_{im}l_{jm}
d_{mm}} \right) \label{lformula} \ .
\end{align}
This formula allows one to compute $L$ one column at a time, starting from the first column.   We proceed by induction:  for the
first column of $L$, $l_{i1} = \sigma_{i1}/\sigma_{11}$, hence $l_{i1} = 0$.
Assume that for some column $u < c(i)$  we have $l_{i1} = \cdots = l_{iu} = 0$,
then using (\ref{lformula}), 
$$
l_{i, u+1} = \frac{1}{d_{u+1, u+1}} \left(\sigma_{i, u+1} - \sum_{m=1}^{u}{l_{im}l_{u+1,m}d_{u+1, u+1}} \right) = \frac{\sigma_{i, u+1}} {d_{u+1, u+1}} \ , 
$$
which implies $l_{i, u+1} = 0$.
\end{proof}

Proposition \ref{propzero} states that a Cholesky factor with banded rows of arbitrary band length (by band length $k_i$ of row $i$ we mean that $k_i$ is the smallest integer such that $l_{ij} = 0$ for all 
$j <i - k_i$) corresponds to a covariance matrix with banded rows of the same 
band lengths.  In particular, the Cholesky factor $L$ is $k$-banded  
if and only if the covariance matrix itself is $k$-banded.  An analogous result holds for the inverse covariance matrix $\Omega$, with rows replaced by columns.
\begin{proposition} \label{invconstraint}
For a positive definite matrix $\Omega$ with modified Cholesky decomposition $T^{T} D^{-1} T = \Omega$, where $T$ is lower triangular, 
for any column $j$ and $r(j) > j$, $\omega_{p,j} = \cdots = \omega_{r(j), j} = 0$ if and only if $t_{p,j} = \cdots = t_{r(j), j} = 0$. 
\end{proposition}

The proof of Proposition \ref{invconstraint} is similar to that of Proposition \ref{propzero} and is omitted.
Proposition \ref{invconstraint} states that the modified Cholesky factor of the inverse $T$ with arbitrary 
column band lengths corresponds to an inverse covariance matrix $\Omega$ with the same column band lengths, and thus an inverse covariance matrix is $k$-banded if and only if its Cholesky factor is $k$-banded.

With these propositions we can investigate maximum likelihood properties of 
Cholesky and inverse Cholesky banding.

\begin{proposition} \label{invcmle}
Banding the modified Cholesky factor $T$ of the inverse covariance matrx $\Omega$ maximizes the normal likelihood subject to the banded constraint,  $\omega_{ij} = 0$ for $|i-j| > k$.
\end{proposition}

\begin{proof}
Let $\Omega_{(k)}$ be a symmetric positive definite matrix with 
$k$ non-zero main sub-diagonals, i.e., $\omega_{(k)ij} = 0$ for $|i - j| > k$. 
The negative normal log-likelihood of $\V{x}_1, \dots, \V{x}_n \ \sim \ N(\V{0}, \Omega_{(k)}^{-1})$, up to a constant, is given by,
$$
f(\Omega_{(k)}) = \tr (\hat\Sigma \Omega_{(k)}) - \log |\Omega_{(k)}|,
$$
where $f$ is a function of the non-zero unique parameters in $\Omega_{(k)}$.  
The $k$-banded constrained maximum likelihood estimator $\hat\Omega_{(k)}$ satisfies $\nabla f(\hat\Omega_{(k)}) = 0$.  Let  $T_{(k)}^{T} D_{(k)}^{-1}T_{(k)} =  \Omega_{(k)}$ be the modified Cholesky decomposition of $\Omega_{(k)}$.  
By Proposition 
\ref{invconstraint}, $t_{(k)ij} = 0$ for $|i-j| > k$.  Let $g(T_{(k)}, D_{(k)}) \equiv f(T_{(k)}^{T} D_{(k)}^{-1}T_{(k)})$, where $g$ is a function of non-zero unique parameters in $(T_{(k)}, D_{(k)})$.  

We continue by establishing that if $\nabla g(\hat T_{(k)}, \hat D_{(k)}) = 0$ then
$\hat T_{(k)}^{T} \hat D_{(k)}^{-1} \hat T_{(k)} =  \hat\Omega_{(k)}$.  Let  
$h(T_{(k)}, D_{(k)}) = T_{(k)}^{T} D_{(k)}^{-1}T_{(k)}$.  Denote the differential of $h$ in the direction $u = (A_{T}, A_{D})$
 evaluated at $(T_{(k)}, D_{(k)})$, by $\nabla h (T_{(k)}, D_{(k)})[u]$.  Then
 \begin{equation}
 \nabla h (T_{(k)}, D_{(k)}) [u]
 = T_{(k)}^{T} D_{(k)}^{-1} A_{T} + A_{T}^{T} D_{(k)}^{-1}T_{(k)} - T_{(k)}^{T} D_{(k)}^{-2} A_{D}T_{(k)} \ ,
 \end{equation}
 where $A_{T}$ is written as a $p \times p$ matrix with non-zero 
entries in the same positions as the
 non-zero lower triangular entries in $T_{(k)}$, and $A_{D}$ is written as a $p\times p$ diagonal matrix.
 Since the diagonal entries of $T _{(k)}$ are all equal to 1 and the diagonal entries of $D_{(k)}$ are positive, 
 one can show by induction that $\nabla h (T_{(k)}, D_{(k)}) [u] = 0$ implies $
u = 0$.  By the chain rule, we have that
 $$
 \nabla g(T_{(k)}, D_{(k)}) [u] = \nabla f(T_{(k)}^{T}D_{(k)}^{-1}T_{(k)})[u] 
\cdot \nabla h (T_{(k)},D_{(k)}) [u] \ .
 $$
 Since $f$ is convex with unique minimizer $\hat\Omega_{(k)}$ it follows that $\nabla f(T_{(k)}^{T}D_{(k)}^{-1}T_{(k)})[u] = 0$ if and only if $T_{(k)}^{T}D_{(k)}^{-1}T_{(k)} = \hat\Omega_{(k)}$ unless $u =0$. 
Hence we have that $\nabla g(T_{(k)}, D_{(k)}) [u] = 0$ iff $\nabla f(T_{(k)}^{T}D_{(k)}^{-1}T_{(k)})[u] = 0$  and  
$\hat T_{(k)}^{T} \hat D_{(k)}^{-1} \hat T_{(k)} =  \hat\Omega_{(k)}$.
 
Minimizing $g(T_{(k)}, D_{(k)})$, which can be expressed as,
$$
g(T_{(k)}, D_{(k)}) = \sum_{j=1}^{p} {\left( n \log d_{(k)jj} +\sum_{i=1}^{n} {\frac{1}{d_{(k)jj}}  
  \left(x_{ij} + \sum_{v=j-k}^{j-1}{t_{(k)jv} x_{iv}} \right)^2}       \right)}, 
$$
is equivalent to minimizing,
$$
g_{j}(t_{(k)j,j-k}, \ldots, t_{(k)j,j-1}, d_{(k)jj} ) = n \log d_{(k)jj} +\sum_{i=1}^{n} {\frac{1}{d_{(k)jj}}  
  \left(x_{ij} - \sum_{v=j-k}^{j-1}{(-t_{(k)jv}) x_{iv}} \right)^2},
$$
for each row $1 \leq j \leq p$.  For row $j$, the solution 
to $\nabla g_{j}(\hat t_{(k)j,j-k}, \ldots, \hat t_{(k)j,j-1}, \hat d_{(k)jj} ) = 0$, gives exactly the
ordinary least squares regression coefficients (with the opposite sign) 
from regressing  $\V{x}_j$ on $\V{x}_{j-k}, \ldots, \V{x}_{j-1}$, and
the sample variance of the $n$ residuals from this fit.  Thus the solution coincides with the output of the inverse Cholesky banding algorithm.
\end{proof}

Next, we show that banding the Cholesky factor of the covariance matrix itself 
does not give the constrained maximum likelihood estimator.  This is due to the inverse being the natural canonical parameter of the multivariate normal distribution.  

\begin{proposition} \label{counterex}
Banding the modified Cholesky factor $L$ of the covariance matrix 
$\Sigma$ does not maximize the normal likelihood
under the constraint that $\sigma_{ij} = 0$ for $|i-j| > k$.
\end{proposition} 

\begin{proof}
  We show that the first-order necessary condition for optimality is not met 
using $p=3$ variables.  
  Let the function $g$ be the negative normal log-likelihood parameterized by the inverse Cholesky factor $T = L^{-1}$ and $D$,
  which is given up to a constant by, 
  \begin{equation} \label{liketd}
  g(T,D) \equiv \ell(T^{T}D^{-1}T) = \sum_{j=1}^{p} {\left( n \log d_{jj} +\sum_{i=1}^{n} {\frac{1}{d_{jj}}  
  \left(x_{ij} + \sum_{v=1}^{j-1}{t_{jv} x_{iv}} \right)^2}       \right)} 
  \end{equation}
  Consider a $3\times 3$ covariance matrix $\Sigma$ with
  the banding constraint $\sigma_{31}=\sigma_{13} = 0$.  This constraint is
 equivalent  to $l_{31} = 0$ by Proposition \ref{propzero}. 
  The inverse Cholesky factor $T$ in terms of the entries in the Cholesky factor $L$ is given by,
  $$
  T = 
\left( \begin{array}{ccc}
	1 & 0 & 0\\
	-l_{21} & 1 & 0\\
	-l_{31} + l_{32} l_{21} & -l_{32} & 1 \\
\end{array} \right)
  $$
Minimizing the negative log-likelihood subject to $l_{31}=0$ is equivalent
to minimizing the unconstrained function
\begin{align*}
b(l_{21}, l_{32}, D)  =  n \sum_{j=1}^{3} { \log d_{jj} } + \frac{1}{d_{11}} \| \V{x}_{1} \|^2
+ \frac{1}{d_{22}} \| \V{x}_{2}-l_{21}\V{x}_{1} \|^2 
+ \frac{1}{d_{33}} \| \V{x}_{3} + l_{32}l_{21} \V{x}_{1} - l_{32} \V{x}_{2} \|^2 \ .
\end{align*}
Taking the partial derivative of $b$ with respect to $l_{21}$,
$$
\frac{\partial }{\partial l_{21}} b(l_{21}, l_{32}, D)
= \frac{1}{d_{22}} \left(  2 l_{21} \V{x}_{1}^{T}  \V{x}_{1} - 2 \V{x}_{1}^{T}  \V{x}_{2} \right)
 + \frac{1}{d_{33}}  \left(  2 l_{32} \V{x}_{1}^{T}  \V{x}_{3} - 2 l_{32}^{2} \V{x}_{1}^{T}  \V{x}_{2} 
 + 2 l_{21} l_{32}^2  \V{x}_{1}^{T}  \V{x}_{1} \right),
$$
and evaluating at the Cholesky banding solution gives, 
$$
\frac{\partial }{\partial l_{21}} b(\hat l_{21}, \hat l_{32}, \hat D)
 = \frac{2 \hat l_{32} \V{x}_{1}^{T}  \V{x}_{3}}{\hat d_{33}}.
$$
Since $\frac{\partial }{\partial l_{21}} b(\hat l_{21}, \hat l_{32}, \hat D) \neq 0$ with 
probability 1, the Cholesky banding solution 
does not satisfy the first-order necessary condition
for being an optimum of an unconstrained differentiable function $b$, 
and hence Cholesky banding does not maximize
the constrained normal likelihood.
\end{proof}

The constrained maximum likelihood estimator can be computed by the algorithm proposed by \citet{chaudhuri07}, but this algorithm only works for $p < n$.  We are not aware of suitable constrained maximum likelihood estimation algorithms for $p > n$, which makes banding the Cholesky factor a more attractive option for computing a positive definite estimator for large $p$.  In Section \ref{sec:sim}, we briefly compare 
the numerical performance of banding the Cholesky factor to the constrained 
maximum likelihood estimator when $p < n$, and find that the two estimators are in practice very close.

\subsection{The penalized regression approach} 
\label{penols}
Instead of banding the Cholesky factor, more sophisticated regularization approaches can be applied to regressions involved in the computation.  In general,  
for $2 \leq j \leq p$ we can estimate the Cholesky factor by,
\begin{equation} \label{penregress}
 \V{\hat l}_j  = \argmin_{\V{l}_j} \{ \|\V{x}_j - Z_{j} \V{l}_j \|^2 + P_{\lambda} (\V{l}_j) \}.
\end{equation}
Penalty functions $P_{\lambda}$ that encourage sparsity in the coefficient vector 
$\V{l}_j$ are of particular interest.  \citet{huang06} applied the lasso 
penalty in the inverse covariance Cholesky estimation problem, and here we can analogously use 
$$
P_{\lambda}^{L} (\V{l}_j)  = \lambda \sum_{t=1}^{j-1} {|l_{jt}|}.
$$
The lasso penalty function can result in zeros in arbitrary locations in the Cholesky factor, which may or may not lead to any zeros in the resulting covariance matrix.   To impose additional structure, \citet{levina_zhu06} proposed the 
nested lasso penalty, which in our context is given by,
\begin{equation} \label{j0}
P_{\lambda}^{NL} (\V{l}_j) = \lambda \left( |l_{j,j-1}| + 
  \frac{|l_{j,j-2}|}{|l_{j,j-1}|} +
  \frac{|l_{j,j-3}|}{|l_{j,j-2}|} + \cdots +
  \frac{|l_{j,1}|}{|l_{j,2}|} 
\right) \ ,  
\end{equation}
where 0/0 is defined as 0.  This penalty imposes the restriction that $l_{jt} = 0$ if $l_{j,t+1}=0$.  By Proposition \ref{propzero}, this means that all the zeros estimated in the Cholesky factor $\hat L$ will be preserved in $\hat \Sigma$.   This is not the case in the inverse Cholesky decomposition for which this penalty was originally proposed by  \citet{levina_zhu06}, although some (not all) zeros are preserved in that case as  well. 

In practice, \citet{levina_zhu06} recommend using a slightly modified 
version of (\ref{j0}) where the first term is divided by the univariate regression coefficient from regressing $\V{x}_j$ on $\V{e}_{j-1}$ alone, to address a potential difference of scales, which is the version we used in simulations.  
Note that both lasso and nested lasso have much higher computational 
cost than banding, and are not appropriate for very large $p$; however, 
the additional flexibility of the sparsity structure may work 
well in some cases.

\section{Numerical results} \label{sec:sim}
In this section we present a simulation study which compares 
the performance of all the covariance estimators discussed in 
Section \ref{sec:regchol}, banding the sample covariance matrix directly 
\citep{bl06}, and, as a benchmark, the shrinkage estimator of \citet{ledoit03}. The main difference between banding the
sample covariance directly and regularizing the Cholesky factor is the guaranteed positive definiteness of the latter. The Ledoit-Wolf estimator is 
a linear combination of the identity matrix and the sample covariance matrix, where linear
coefficients are estimates of asymptotically optimal coefficients under Frobenius loss; it does not introduce any sparsity.

\subsection{Simulation Settings}
We consider two standard covariance structures for ordered variables,
\begin{enumerate}
\item $\Sigma_1$: $\sigma_{ij} = (7/10)^{|i-j|}$ ; 
\item $\Sigma_2: \sigma_{ij} = \I(i=j)+ (4/10) \I(|i-j|=1) + (2/10) \I(|i-j| = 2)+ (2/10) \I(|i-j| = 3)+ (1/10) \I(|i-j|=4)$.
\end{enumerate}
The AR(1) model $\Sigma_1$ has a dense Cholesky factor while the MA(4) 
model $\Sigma_2$ is a banded matrix with $k=4$, and therefore its Cholesky factor is also 4-banded.  The model $\Sigma_1$ was considered by \citet{bl06}, and 
$\Sigma_2$ by \citet{yuan07}.

We generate $n=100$ training observations and another 100 
independent validation observations from $N_{p}(\V{0}, \Sigma)$. 
The dimensions considered were $p = 30$, 100, 200, 500, and 1000.    
Note that lasso and nested lasso were not run for $p=500$ and $1000$ due to their high computational cost. Tuning parameters were selected by minimizing the Frobenius norm 
($\|M\|_{F}^2 = \sum_{i,j} {m_{ij}^2}$) of the difference between the regularized estimate computed with the training observations and the sample covariance computed with the validation observations. Alternatively, 
one could select tuning parameters using the random-splitting scheme of \citet{bl06}, which we use in the data example in Section \ref{sec:realdata}.  The whole process was repeated 50 times. 

To compare estimators, we used the operator norm loss, also known as the matrix 2-norm ($\|M\|^2 = \lambda_{\max}(MM^{T})$), of the difference between the covariance estimator and the truth, 
 $$
 \Delta(\hat\Sigma, \Sigma) = E \| \hat\Sigma-\Sigma\|.
 $$
We also compute the true positive rate (TPR) and true negative rate (TNR), 
defined as 
\begin{align}
\label{tpr}
{\rm TPR}(\hat\Sigma, \Sigma) & = \frac{\#\{(i,j) : \hat\sigma_{ij} \neq 0 \quad {\rm and} \quad \sigma_{ij} \neq
  0 \}} 
{\#\{(i,j): \sigma_{ij} \neq 0 \}} \ , \\
\label{tnr}
{\rm TNR}(\hat\Sigma, \Sigma) & = \frac{\#\{(i,j) : 
  \hat\sigma_{ij} = 0 \quad {\rm and} \quad \sigma_{ij} = 0 \}} 
{\#\{(i,j): \sigma_{ij} = 0 \}} \ .
\end{align}
Note that the sample covariance has a true positive rate of 1, and 
a diagonal estimator has a true negative rate of 1.  
Additionally we measure eigenspace agreement between the estimate 
and the truth using the measure,  for $q = 1, \dots, p$ 
\begin{equation}\label{kq}
K(q) = \sum_{i=1}^{q} {\sum_{j=1}^{q}{ ( \V{\hat{e}}_{(i)}^{T} \V{e}_{(j)})^2  } },
\end{equation}
introduced  by \citet{krzanowski79}, where  $\V{\hat{e}}_{(i)}$ denotes the estimated eigenvector
corresponding to the $i$-th largest estimated eigenvalue, and
$\V{e}_{(i)}$ the  
true eigenvector corresponding to the $i$-th largest true 
eigenvalue. Note that $K(q) = q$ indicates perfect agreement of the eigenspaces 
spanned by the first $q$ eigenvectors. 

\subsection{Results}
\begin{table}[!ht]
\caption{Operator Norm Loss, average(SE) over 50 replications}
\begin{center}
{\small
\begin{tabular}{|l|cccccc|}
\hline
$p$ & Sample & Ledoit-Wolf & Sample Banding & Cholesky Banding & Lasso & Nested Lasso\\
\hline
& \multicolumn{6}{c|}{$\Sigma_1$}\\
\hline
30 & 1.75(0.04)  & 1.67(0.04)  & 1.27(0.04)  & 1.27(0.03)  & 1.68(0.05)  & 1.45(0.04)  \\
100 & 4.14(0.07)  & 3.06(0.03)  & 1.58(0.03)  & 1.56(0.03)  & 3.50(0.03)  & 1.78(0.03)  \\
200 & 6.55(0.07)  & 3.79(0.02)  & 1.75(0.03)  & 1.74(0.03) & 3.90(0.01)  & 1.93(0.03) \\
500 & 12.57(0.08)  & 4.42(0.01)  & 1.95(0.03)  & 1.91(0.02)  &--&--\\ 
1000 & 20.65(0.09)  & 4.64(0.00)  & 2.08(0.03)  & 2.00(0.02) &--&-- \\
\hline
& \multicolumn{6}{c|}{$\Sigma_2$}\\
\hline
30 & 1.44(0.03)  & 1.13(0.02)  & 0.77(0.02)  & 0.75(0.02)  & 1.23(0.02)  & 0.88(0.02)  \\
100 & 3.34(0.04)  & 1.64(0.01)  & 0.92(0.02)  & 0.89(0.02)  & 1.63(0.01)  & 1.00(0.01)  \\
200 & 5.36(0.04)  & 1.78(0.00)  & 0.99(0.02)  & 0.93(0.02)  & 1.71(0.00)  & 1.08(0.01)  \\  
500 & 10.36(0.05)  & 1.84(0.00)  & 1.09(0.02)  & 1.05(0.02)  &--&--\\
1000 & 17.60(0.07)  & 1.85(0.00)  & 1.19(0.02)  & 1.14(0.02)  &--&--\\
\hline
\end{tabular}}
\label{machold}
\end{center}
\end{table}

The averages and standard errors over 50 replications of the 
operator norm loss for both models are given in Table \ref{machold}.  
One can see that banding the Cholesky 
factor provides the best performance in every case.  It 
outperforms banding the sample covariance directly, particularly in high 
dimensions, and both banding methods outperform 
the Ledoit-Wolf estimator as well as both regularized regression methods.  

The banded maximum likelihood estimator was also computed using the algorithm 
of \citet{chaudhuri07} for $p = 30$ (the algorithm is only applicable when $p < n$).  Its loss values are 1.27(0.04) for $\Sigma_1$  and 0.76(0.02) for 
$\Sigma_2$, which are essentially the same as those for Cholesky banding 
for $p=30$.  As expected, the margin by which sparse regularized estimators 
outperform non-sparse estimators (the sample and Ledoit-Wolf) is larger 
for the sparse population covariance $\Sigma_2$.

For the sparse matrix $\Sigma_2$, we also report true positive and true negative rates of estimating zeros  in Table \ref{cholsparsity}.  These rates also depend on the tuning parameter ($k$ for banding and $\lambda$ for 
the lasso and nested lasso).  
Both Cholesky banding and sample covariance banding have perfect true negative rates, meaning that all of the
realizations had at most 4 non-zero sub-diagonals.  We see a better true positive rate for banding the Cholesky factor than for banding the sample covariance matrix, which means that banding the sample tends to set more diagonals to zero than necessary.   This is partly because the entries on the fourth sub-diagonal of $\Sigma_2$ are quite small.    
 The lasso method has a low true negative rate, which is expected since zeros in the Cholesky factor are not preserved, and the nested lasso does reasonably well on both but not as well as Cholesky banding.

\begin{table}[!ht]
\caption{True Positive/True Negative Ratea for $\Sigma_2$ (\%), average(SE) over 50 replications}
\begin{center}
{\small
\begin{tabular}{|l|cccc|}
\hline
$p$ & Banding & Cholesky Banding & Lasso & Nested Lasso \\
\hline
30 &  88.18(1.69)  / 100(0)   &  91.00(1.78)  / 100(0)   & 99.71(0.08)  / 3.86(0.40)  & 93.89(0.75)  / 88.9(0.88)  \\
100 & 88.68(1.75)  / 100(0)   & 94.09(1.50)  / 100(0)   & 90.69(0.27)  / 36.45(0.40)  & 94.44(0.36)  / 97.07(0.16)  \\ 
200 & 88.59(1.77)  / 100(0)   & 95.04(1.42)  / 100(0)  & 90.76(0.18)  / 34.63(0.28)  & 94.01(0.32)  / 98.72(0.05)  \\
500 & 88.04(1.78)  / 100(0)   & 96.01(1.31)  / 100(0) & --& -- \\
1000& 87.02(1.78)  / 100(0)   & 96.51(1.24)  / 100(0) & -- & --  \\ 
\hline
\end{tabular}}
\label{cholsparsity}
\end{center}
\end{table}

\begin{figure}
\twoImages{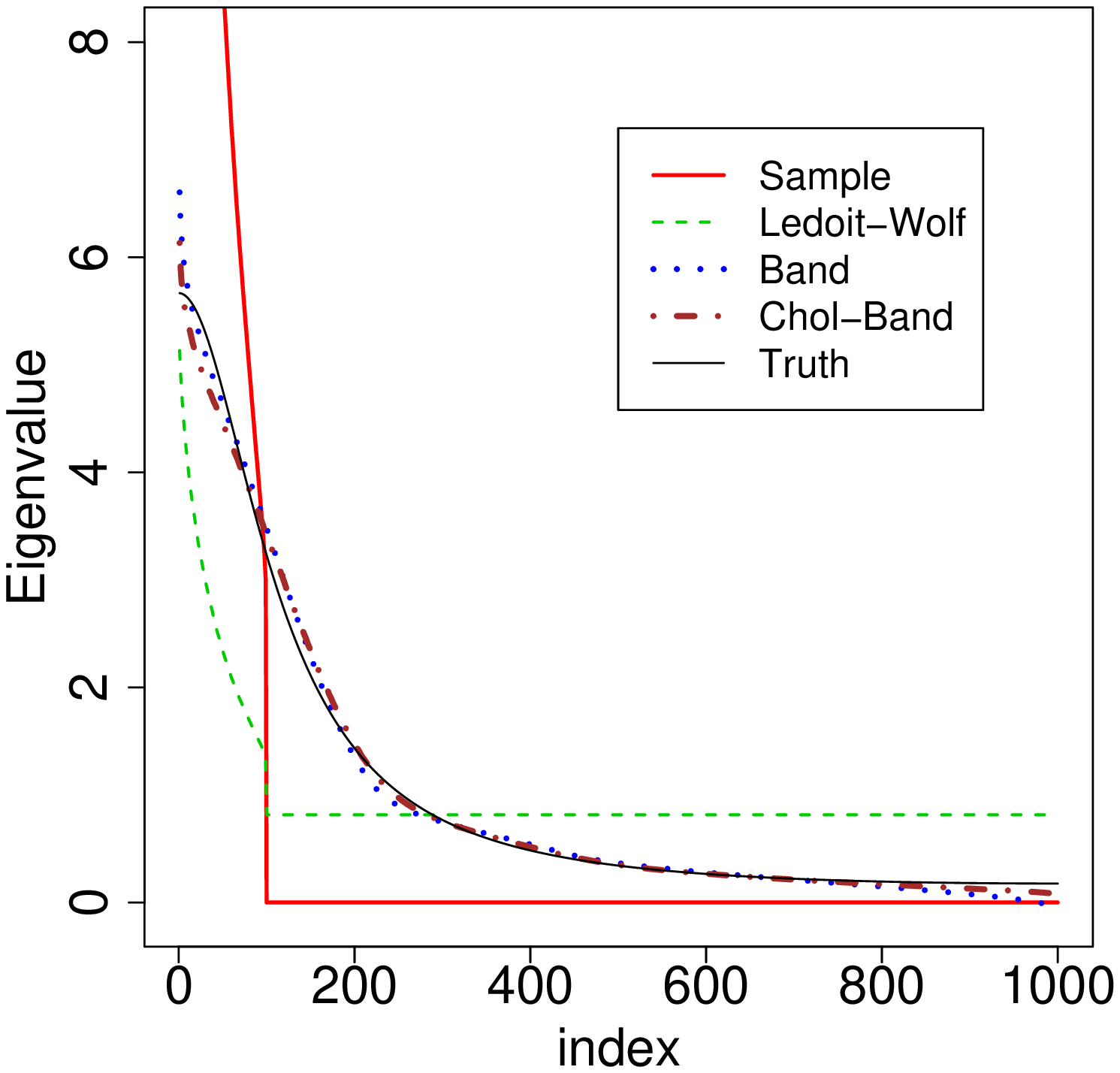}{2.7in}{$\Sigma_1$}{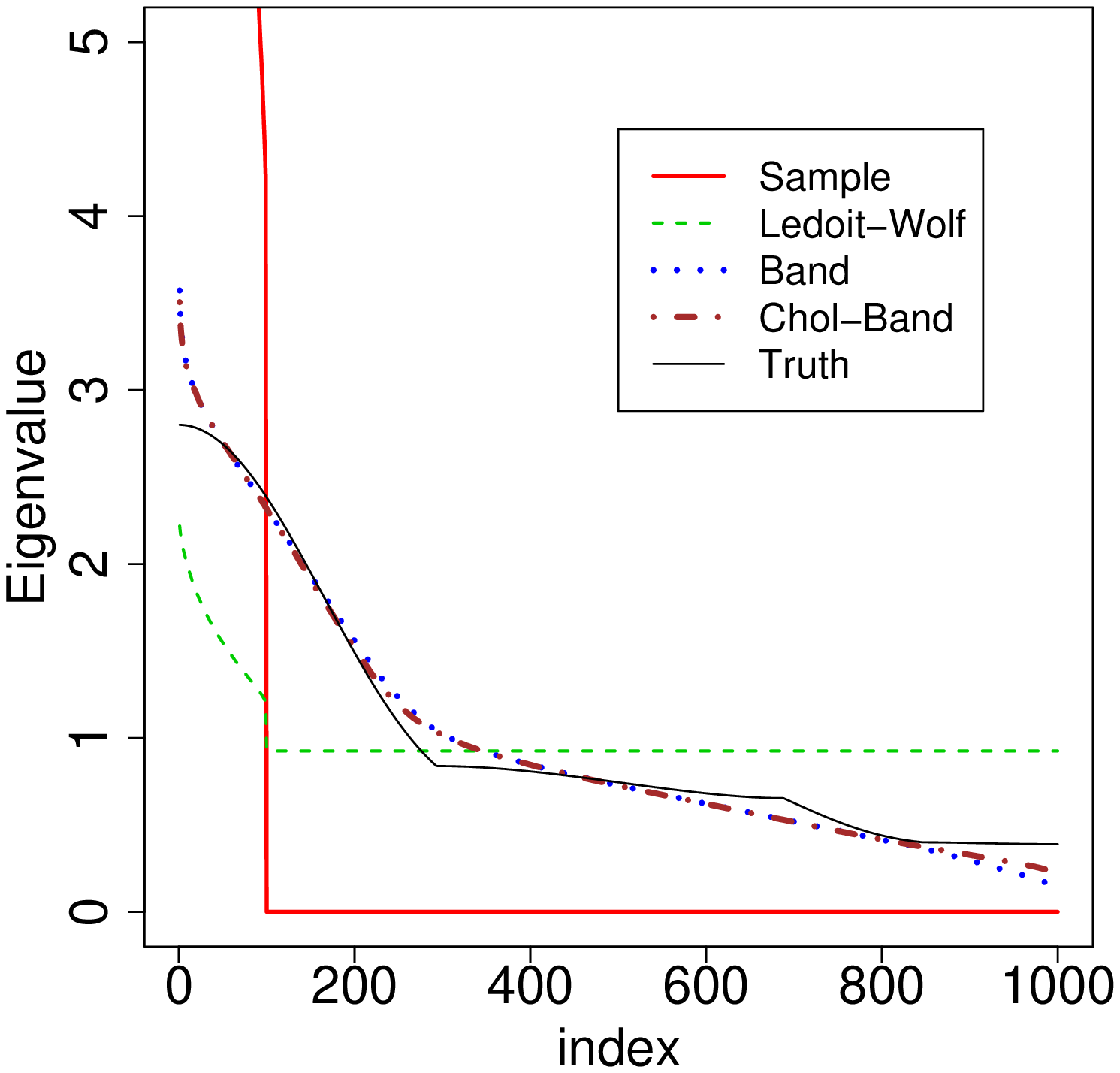}{2.7in}{$\Sigma_2$}
\caption{Scree plots (averaged over 50 replications) for 
$p=1000$. } \label{fig:scree}  
\end{figure}

In Figure \ref{fig:scree} we plot the averaged estimated eigenvalues in descending
order for sample banding, Cholesky banding, the sample covariance,
and the Ledoit-Wolf estimator, as well as the true eigenvalues, for both models and $p=1000$.  Since $n=100$, the sample covariance matrix only has 99 non-zero eigenvalues.  Cholesky banding and sample banding perform similarly 
for both models, with Cholesky banding having a slight edge for the small eigenvalues.  The banding methods outperform both the sample covariance and the Ledoit-Wolf estimator
by a considerable amount, especially for larger true eigenvalues. 

Since sample covariance banding does not necessarily produce a positive definite estimator, we also report the percentage of estimates that are positive definite in Table \ref{posdef}.  We see that for the dense matrix $\Sigma_1$, sample banding has 0 out of 50 positive definite realizations for $p \geq 200$; for the sparse matrix $\Sigma_2$, sample banding has 50 out of 50 positive definite realizations for $p \leq 200$, 49 for 
for $p = 500$ and 48 for $p=1000$; it is clear that, for both models, the larger $p$, the harder it is to keep positive definiteness. 

{\small
\begin{table}[!ht]
\caption{Percentage of banded sample covariance realizations that are positive definite (based on 50 replications)}
\begin{center}
\begin{tabular}{|c||c|c|c|c|c|}
\hline
 & \multicolumn{5}{c|}{$p$} \\
\cline{2-6}
Model & 30 & 100& 200 & 500& 1000 \\
\hline
$\Sigma_1$ &  66 & 8 & 0 & 0 & 0\\
\hline
$\Sigma_2$ & 100 & 100 & 100 & 98 & 96 \\
\hline
\end{tabular}
\label{posdef}
\end{center}
\end{table}}

Finally, Figure \ref{fig:kq} shows a plot of the averaged eigenspace agreement measure $K(q)$ versus $q$, for $p=1000$ variables, along with the line $K(q) = q$ representing perfect eigenspace agreement. We see that both Cholesky banding and ordinary banding perform roughly the same under this measure; both outperform
the sample covariance matrix and 
the Ledoit-Wolf estimator, which have the same eigenvectors, since  the Ledoit-Wolf estimator is a linear combination of the sample covariance and the identity.

\begin{figure}[ht!]
\twoImages{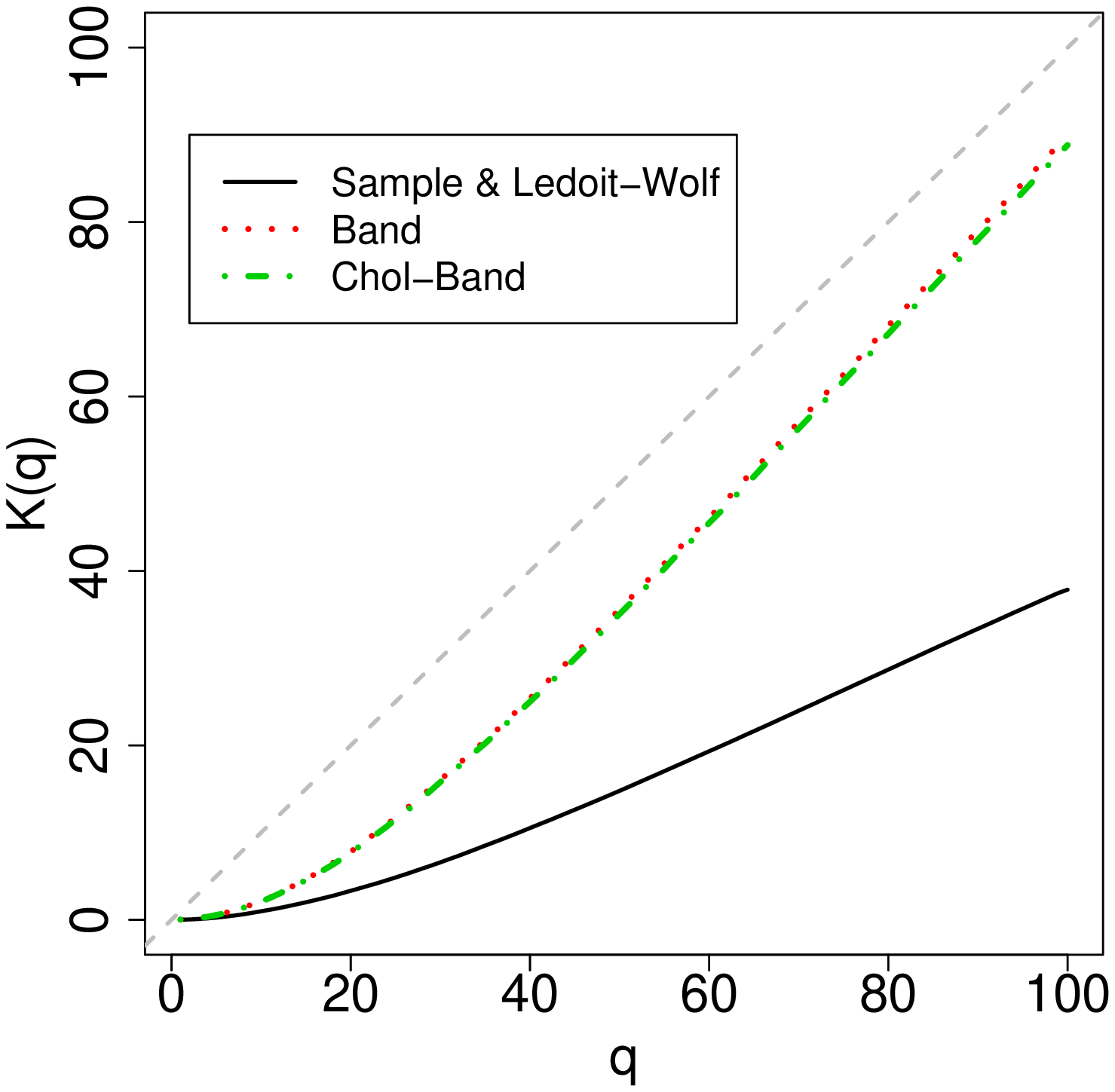}{2.7in}{$\Sigma_1$}{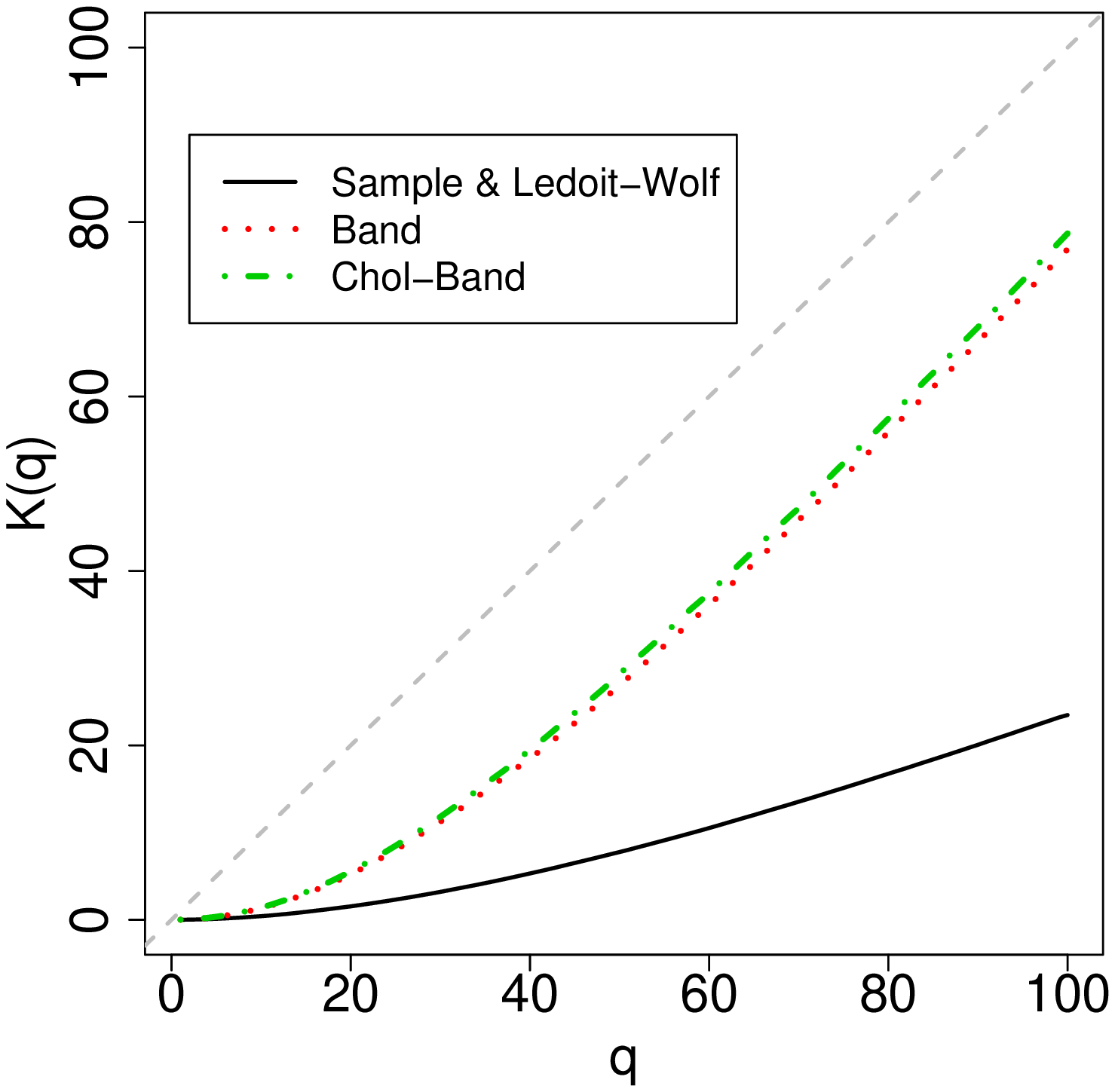}{2.7in}{$\Sigma_2$}
\caption{\label{fig:kq} $K(q)$ versus $q$ (averaged over 50 replications) for $p=1000$.  $K(q)=q$ corresponds to perfect agreement.}  
\end{figure}

\section{Sonar data example} \label{sec:realdata}
In this section we illustrate the effects of Cholesky banding and 
sample covariance banding on SONAR data from the UCI machine learning 
data repository \citep{uci}.  This dataset has 111 spectra from metal cylinders and 97 spectra from rocks, where each 
spectrum has 60 frequency band energy measurements. 
 These spectra were measured at multiple angles for the same objects, but following previous analyses of the dataset we assume independence of the spectra.

\begin{figure}[ht!]
\twoImages{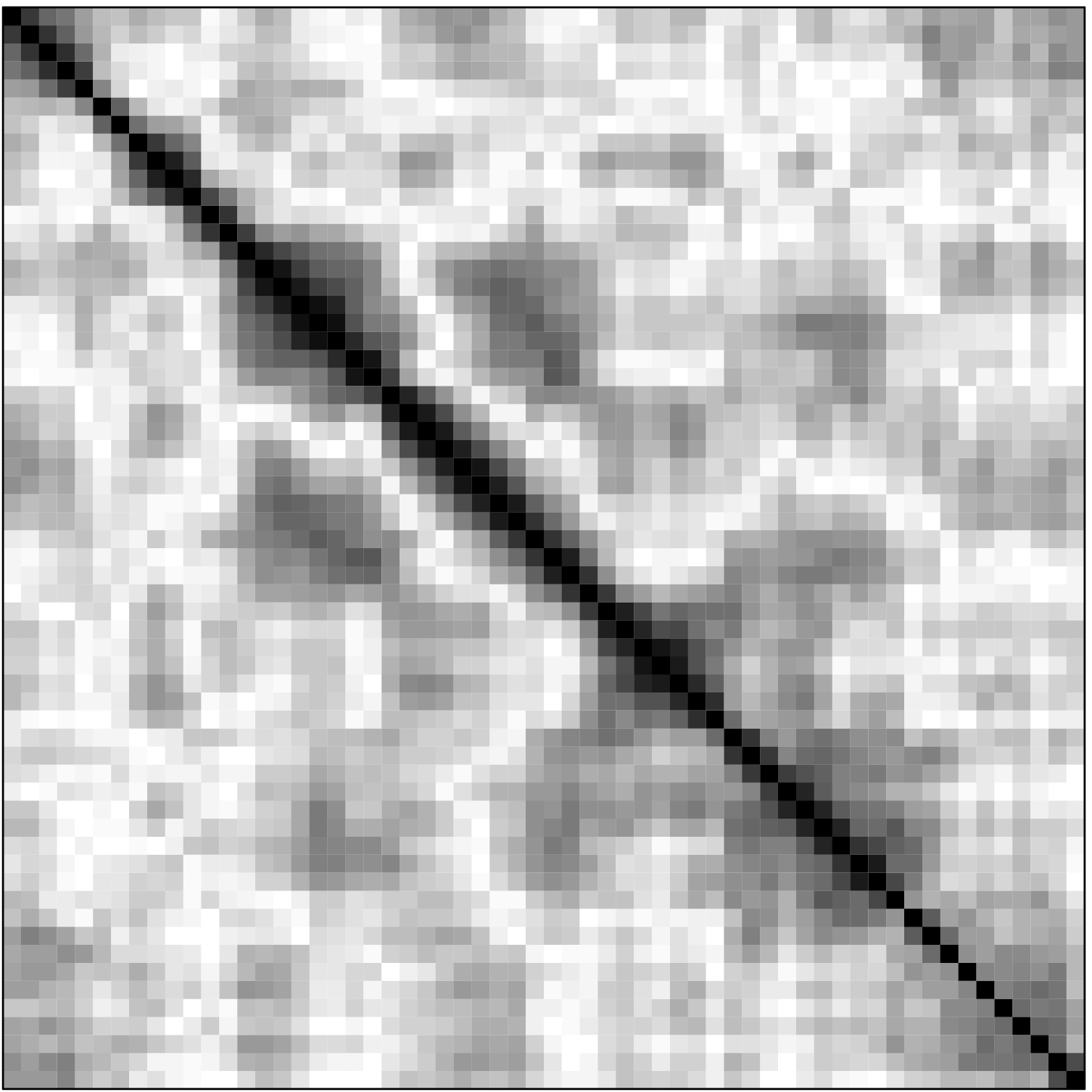}{2.2in}{Samp. Metal}{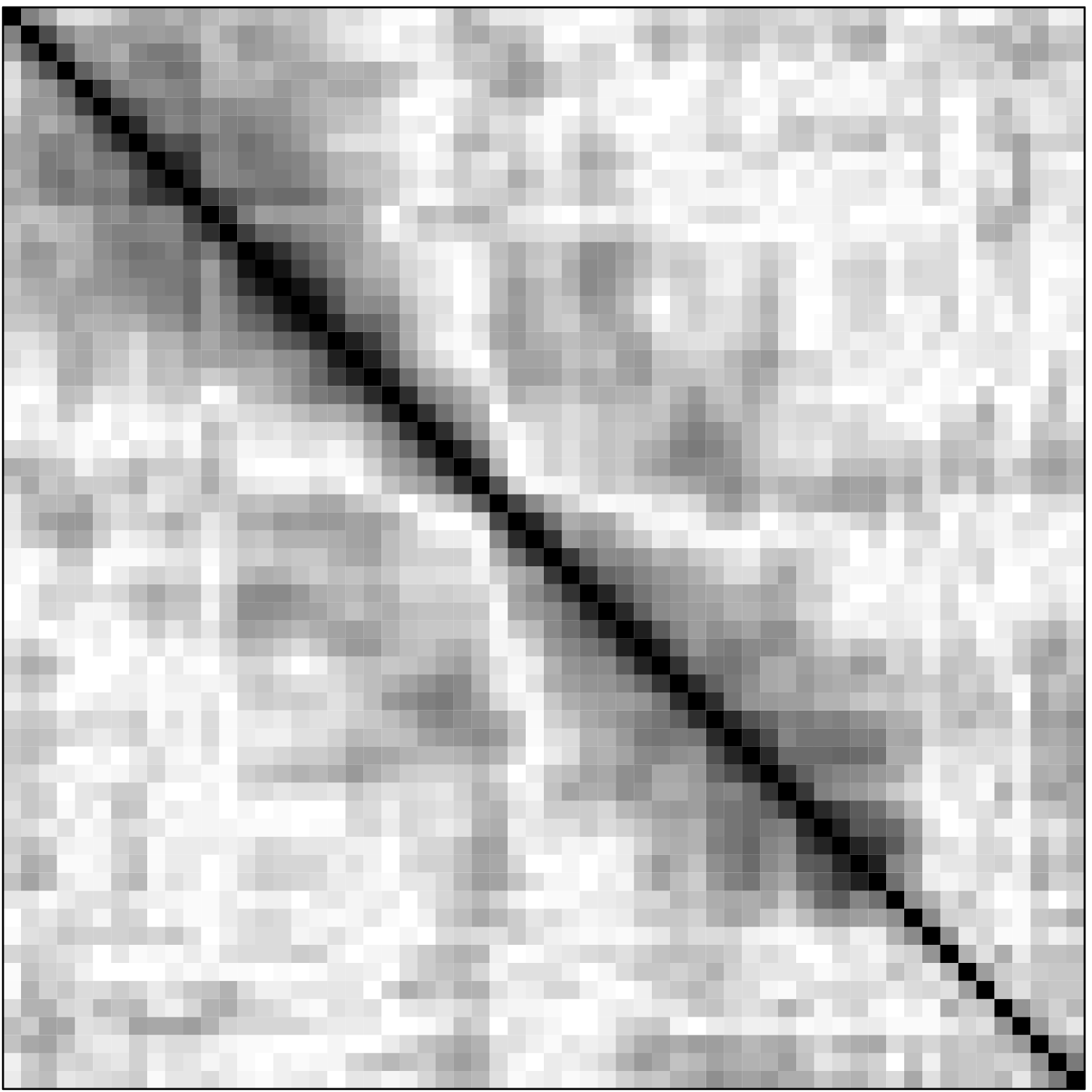}{2.2in}{Samp. Rock}
\twoImages{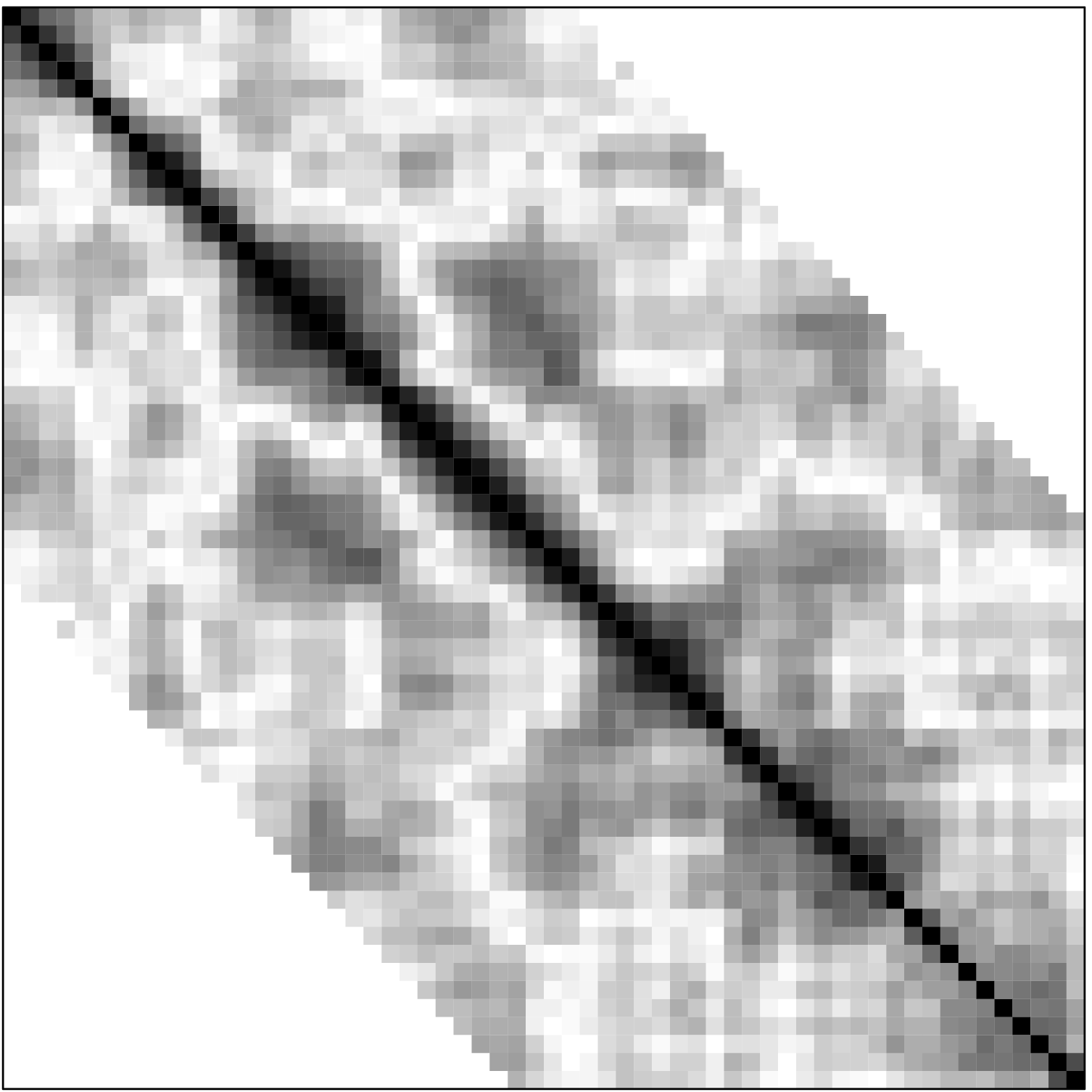}{2.2in}{Samp. Band. $k=31$ Metal}{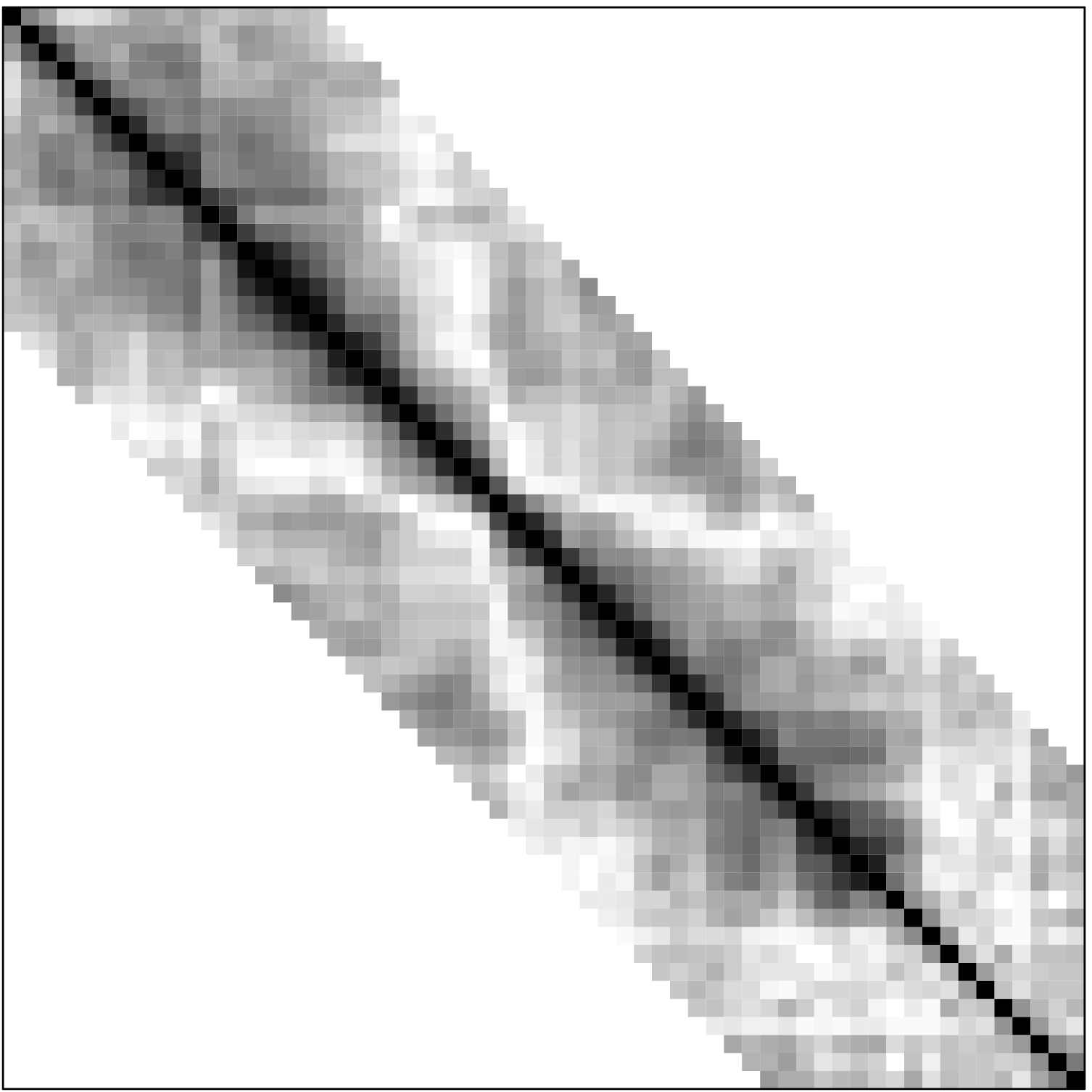}{2.2in}{Samp. Band. $k=17$ Rock}
\twoImages{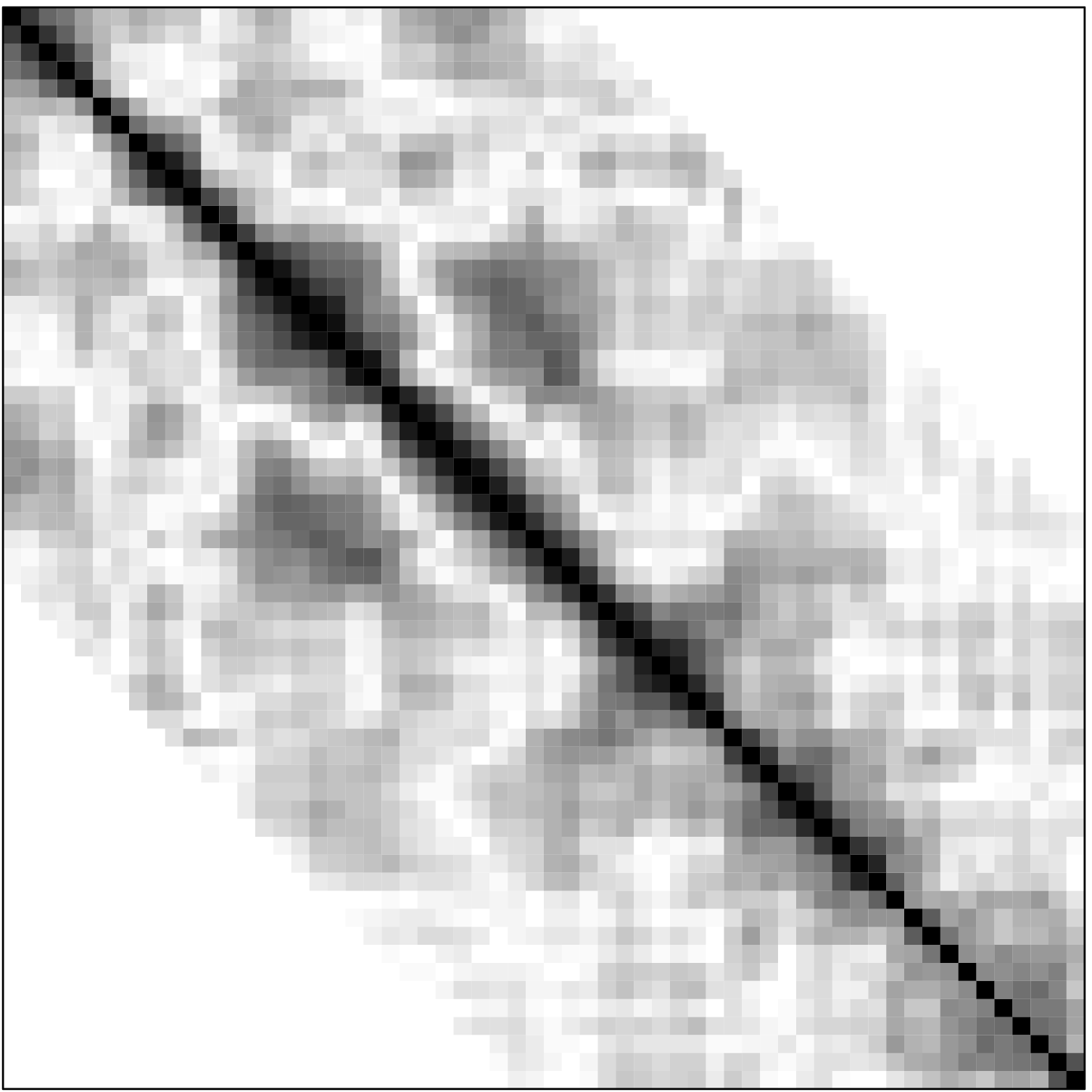}{2.2in}{Chol. Band. $k=31$ Metal}{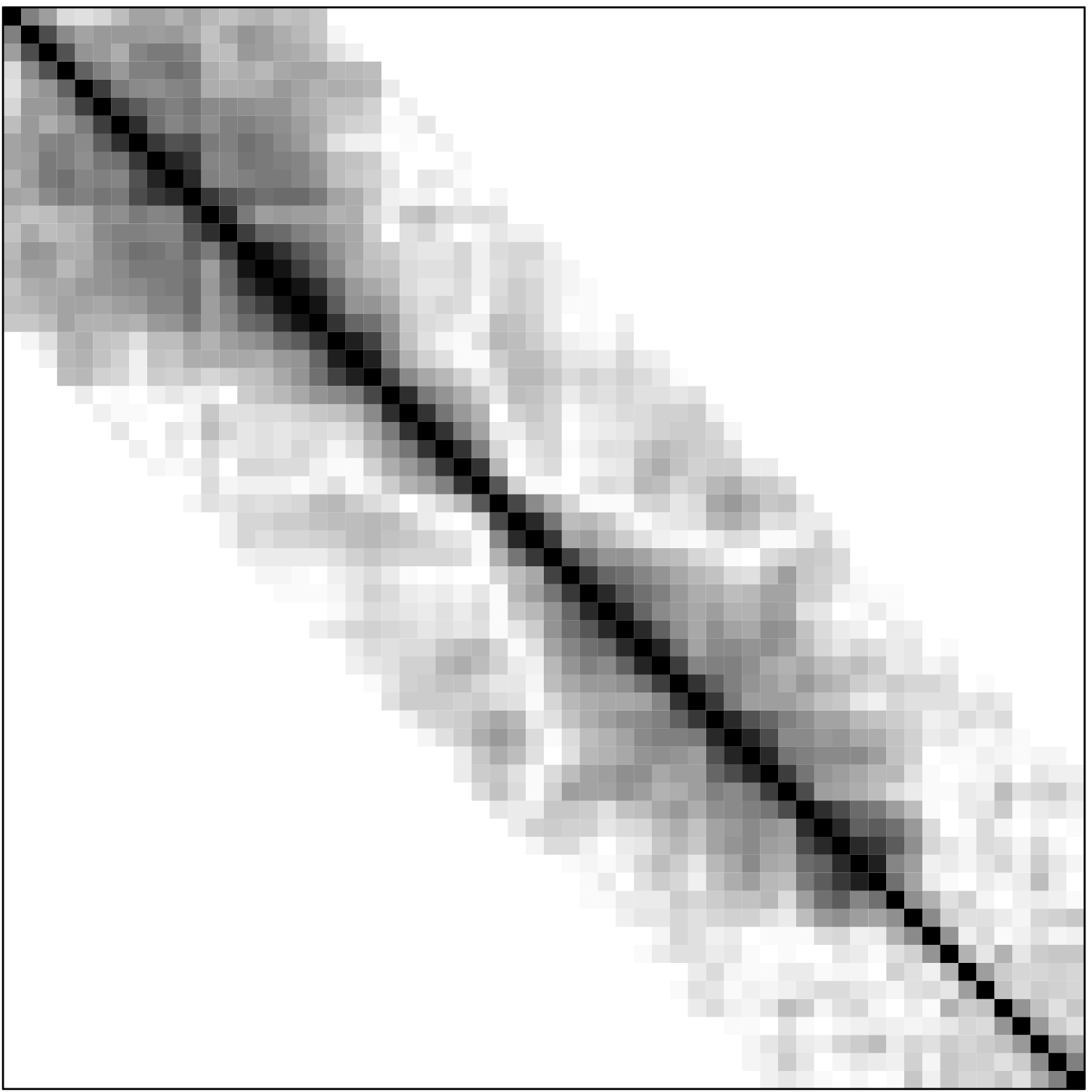}{2.2in}{Chol. Band. $k=17$ Rock}
\caption{\label{fig:heatest} Heatmaps of the absolute values of the correlation estimates.  White is magnitude 0 and black is magnitude 1.}  
\end{figure}

\begin{figure}[ht!]
\twoImages{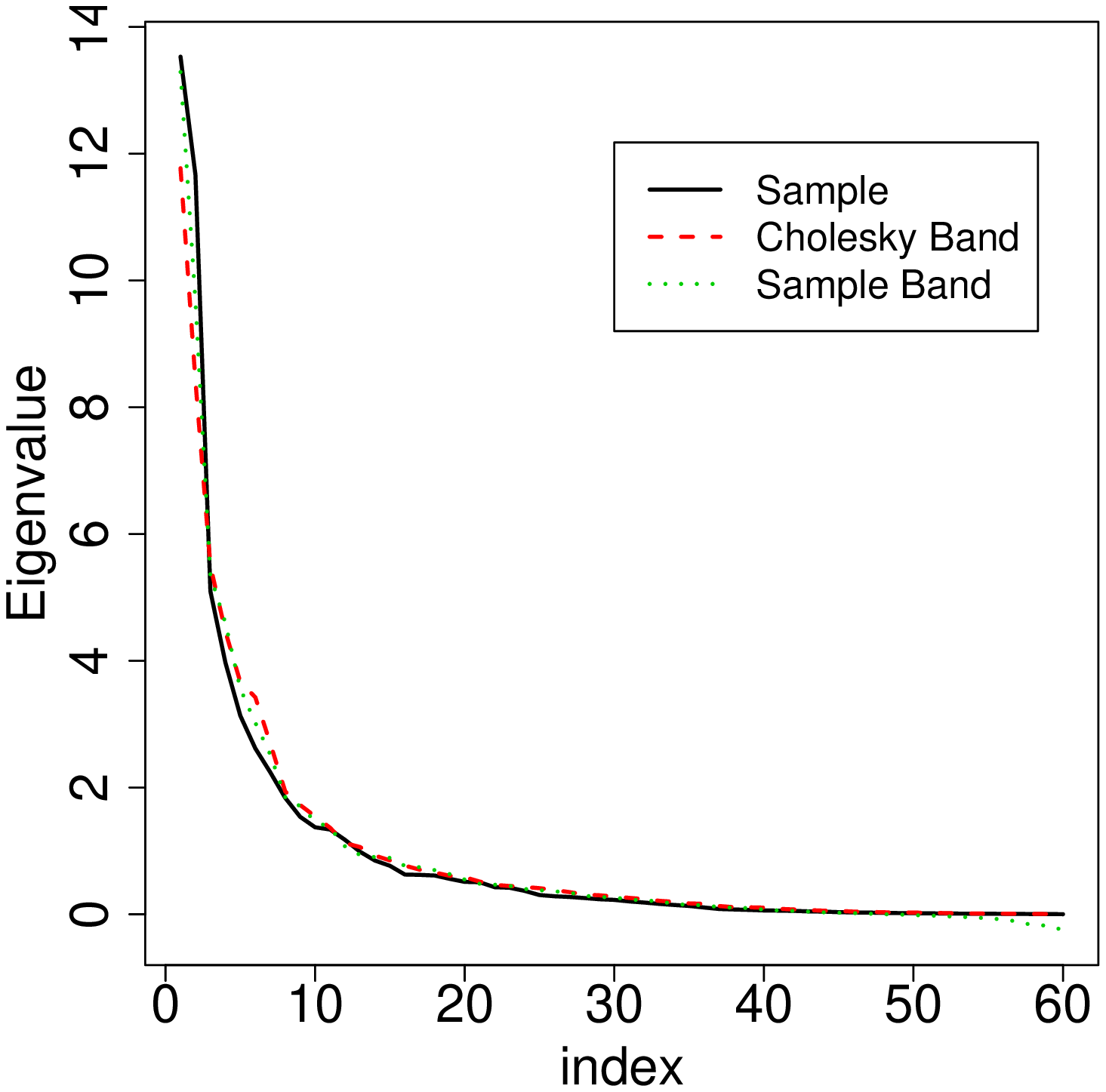}{2.7in}{(a)}{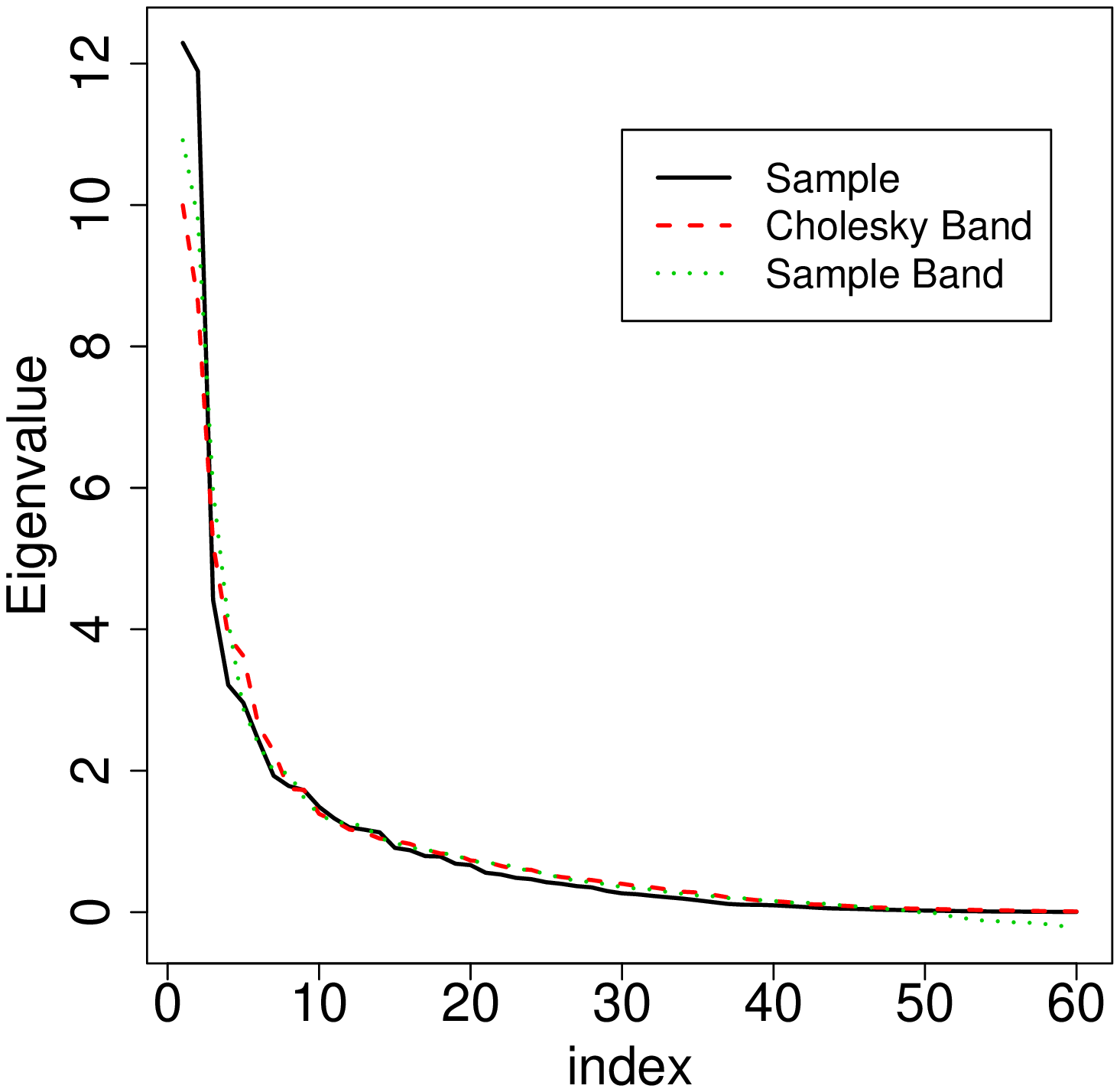}{2.7in}{(b)}
\twoImages{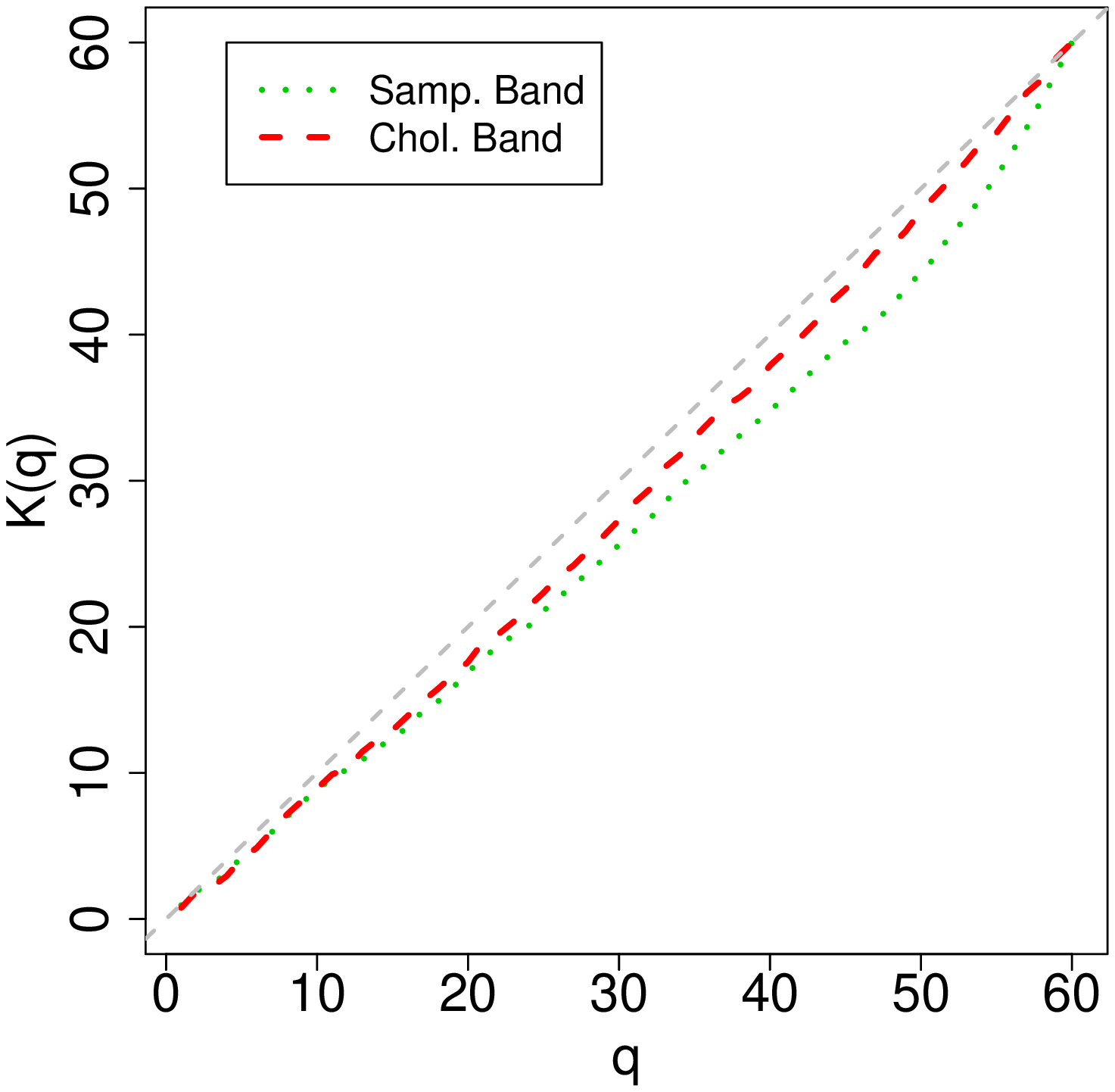}{2.7in}{(c)}{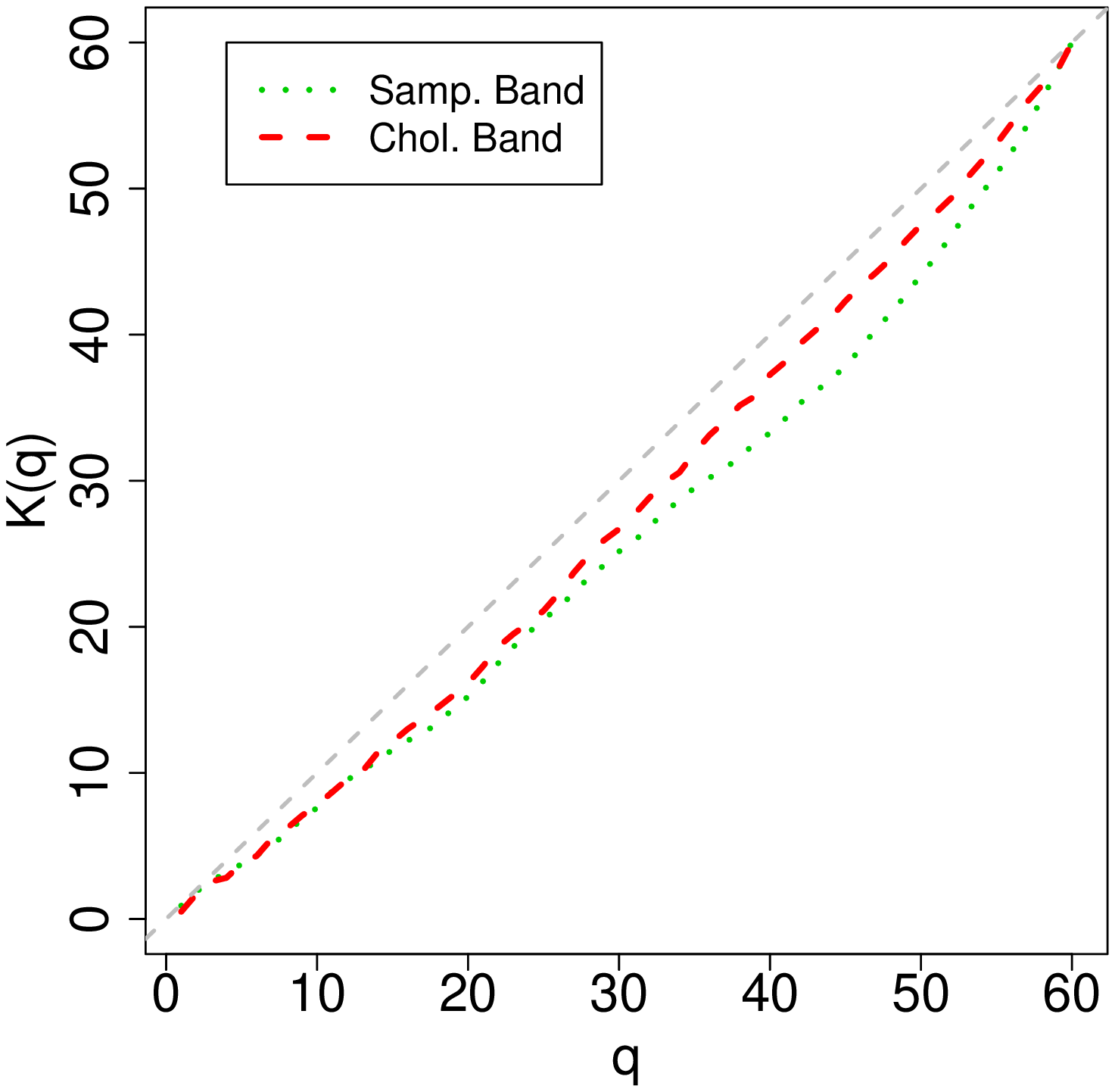}{2.7in}{(d)}
\caption{\label{fig:estscree} (a) Scree plots of the banded estimators, using the metal spectra;  (b) Scree
plots using the rock spectra; (c) Eigenspace agreement with the sample covariance matrix using the metal spectra.
(d) Eigenspace agreement with the sample covariance matrix using the rock spectra.  Note that in (c) and (d), $K(q)=q$, drawn
as the gray dashed line, corresponds
to perfect agreement, see (\ref{kq}) for the definition of $K(q)$.}  
\end{figure}

The top panel of Figure \ref{fig:heatest} shows heatmaps of the absolute 
values of the sample correlation matrices for metal and rock (we standardize the variables first to facilitate comparison for metal and rock spectra, which are on different scales).  Both matrices show a general pattern of correlations decaying as one moves away from the diagonal, which makes banding a reasonable option. 
 
The banding parameter $k$ for both banding methods was selected using the random-splitting
scheme of \citet{bl06},  
$$
\hat k = \argmin_k \frac{1}{N} \sum_{v=1}^{N} {\|\hat\Sigma^{(v)}_{(k)} - \tilde\Sigma^{(v)} \|_{F}} \ ,
$$
where $\hat\Sigma^{(v)}_{(k)}$ is the banded estimator with $k$ bands   
computed on the training data, and 
$\tilde\Sigma^{(v)}$ is the sample covariance of the validation data.  
To obtain these training and validation sets, the data was split at random  
$N=100$ times, with 1/3 of the sample used for training.  
For metal, 
Cholesky banding and sample banding both chose $\hat k = 31$ sub-diagonals; 
for rock, Cholesky banding chose $\hat k = 17$ and sample banding chose $\hat k = 18$.  Since these values are so close,  for easier visual comparison 
we show Cholesky banding and sample banding both computed with $\hat k = 17$ for 
the rock spectra.  The heatmaps of the absolute values of the banded 
estimators are shown in Figure \ref{fig:heatest}.
We see that Cholesky banding shrinks the non-zero correlations whereas the 
sample banding does not, which is the property that allows Cholesky banding to achieve positive definiteness.   

We also show eigenvalue plots for these estimators in 
Figure \ref{fig:estscree}(a) and (b), and the eigenspace agreement measure 
between the banded estimators and the sample covariance in Figure \ref{fig:estscree}(c) and (d), using the agreement measure (\ref{kq}).
We see that the sample covariance has the most spread out eigenvalues,
and the eigenvalues from Cholesky banding have the least spread, as we would expect.  For eigenvectors, there are no major differences between the estimators, a result consistent with simulations.

We also compared the performance of the various estimators if they are used in quadratic discriminant analysis (QDA) to discriminate between rock and metal. An observation $\V{x}$ is classified as rock ($k=0$) or metal ($k=1$) using the QDA rule,
$$
G(\V{x}) = \argmax_{k} \left\{\frac{1}{2} \log |\hat\Omega_{k}| - \frac{1}{2}(\V{x}-\V{\hat\mu}_{k})^{T}\hat\Omega_{k} (\V{x}-\V{\hat\mu}_{k}) + \log \hat\pi_{k} \right\},
$$
where $\hat\pi_{k}$ is the proportion of class $k$ observations in the training sample, $\V{\hat\mu}_{k}$ is the training class $k$ sample mean vector, and $\hat\Omega_{k}$ is the inverse covariance estimate computed with the class $k$ training observations.  A full description of QDA can be found in \citet{mardia79}.
In addition to banding the Cholesky factor of covariance and of the inverse, we also added a diagonal estimator of the covariance matrix (which corresponds to the 
naive Bayes classifier).  Leave-one-out cross validation was used to
estimate the testing error, and the banding parameters were selected with 10 
random splits with 1/3 of the data used for training, using Frobenius loss for covariance Cholesky banding and the validation likelihood for the inverse 
covariance Cholesky banding.  Banding the sample covariance was omitted because its lack of positive definiteness led to inversion problems.   
The test errors (\%) were  24.0(3.0) for the sample covariance, 
32.7(3.3) for naive Bayes,  20.2(2.8) for covariance Cholesky banding, 
and 14.9(2.5) for inverse Cholesky banding.  Both banding methods are substantially better than either estimating the whole dependency structure by the sample covariance or not estimating it at all (naive Bayes), and the inverse Cholesky 
banding does better in this case because it introduces sparsity directly 
in the inverse covariance.

\section{Summary and discussion}
In this paper we proposed a new regression interpretation of the Cholesky 
factor of the covariance matrix, which was previously only available for the Cholesky factor of the inverse.  Banding of this Cholesky factor gives a banded positive definite estimator of the covariance, unlike banding the sample covariance matrix, and was shown to perform better numerically.  An attractive property of the 
banded Cholesky estimator is its low computational cost, the same as that of banding the sample covariance matrix itself, and thus there is no computational penalty to pay for enforcing positive definiteness.  More complicated regularization obtained from penalties such as the lasso or the nested lasso 
can be applied using the same regression interpretation, 
but at an additional computational cost.  The proposed estimators 
perform well numerically under a variety of measures.

We also connected sparsity in banded Cholesky factors with sparsity in the covariance matrix and the inverse covariance matrix, which allows us to show that 
inverse Cholesky banding is equivalent to constrained maximum likelihood under the banded constraint.  Banding the Cholesky factor of the covariance itself is not equivalent to constrained maximum likelihood, but we found empirically they perform similarly.    In terms of convergence rates, one would expect a convergence result analogous to the one for inverse Cholesky banding established by \citet{bl06} to hold here as well, but this case presents substantial extra technical difficulties in analysis, due to the fact that the errors used as predictors in 
the regressions required to compute the Cholesky factor are unobservable and have to be estimated by residuals.  Nonetheless, we expect the method to be equally useful based on its good practical performance.  

\section*{Acknowledments}
We thank Richard Davis (Columbia) for pointing out the use of
regression on residuals in time series, and Bala Rajaratnam
(Stanford) for helpful discussions on sparse Cholesky factors.
A.J. Rothman's research is supported in part by the Yahoo!
Ph.D. Fellowship.  E. Levina's research is supported in part by
grants from the NSF (DMS-0505424, DMS-0805798).  J. Zhu's
research is supported in part by grants from the NSF (DMS-0705532
and DMS-0748389).

\bibliography{allref}
 
\end{document}